\documentclass[proceedings]{dmtcs} 
\usepackage[utf8]{inputenc}
\usepackage[cmex10]{amsmath}
\usepackage{hyperref}
\pdfoutput=1 
\usepackage{algorithmic}

\usepackage{subfigure}
\usepackage[round]{natbib}
\hypersetup{
  pdfinfo={
    Title={On Leader Green Election},
    Author={J. Cicho{\'n}, R. Kapelko and D. Markiewicz},
  }
}

\usepackage{natbib}
\newtheorem{definition}{Definition}
\newtheorem{lemma}{Lemma}
\newtheorem{theorem}{Theorem}
\newtheorem{corollary}{Corollary} 

\newcommand{\ZZ}{\mathbb{Z}}

\newcommand{\RR}{\mathbb{R}}
\newcommand{\CC}{\mathbb{C}}

\newcommand{\E}[1]{\mathbf{E}\left[#1\right]}
\newcommand{\Geo}[1]{\mathrm{Geo}(#1)}
\newcommand{\MGeo}[2]{\mathrm{MGeo}(#1,#2)}
\newcommand{\WGeo}[2]{\mathrm{WMGeo}(#1,#2)}

\newcommand{\Res}[3]{\mathrm{Res}(#1:#2=#3)}

\newcommand{\HarmonicN}[1]{\mathrm{H_{#1}}}

\newcommand{\BigO}[1]{\mathrm{O}\left(#1\right)}

\newcommand{\card}[1]{\mathrm{card}\left(#1\right)}

\newcommand{\MSP}[2]{\mathrm{MSP}\left(#1,#2\right)}

\title[On Leader Green Election]{On Leader Green Election}

\author[J. Cicho{\'n}, R. Kapelko and D. Markiewicz]{Jacek Cicho{\'n}\thanks{This paper was supported by Polish National Science Center (NCN) grant number 2013/09/B/ST6/02258}
\and Rafal Kapelko \and Dominik Markiewicz}

\address{Department of Computer Science\\
Faculty of Fundamental Problems of Technology\\
Wroc{\l}aw University of Technology\\
Poland
}

\begin{document}
\label{firstpage}
\maketitle
\begin{abstract}
We investigate the number of survivors in the Leader Green Election (LGE) algorithm introduced by
P. Jacquet, D. Milioris and P. M{\"u}hlethaler in 2013.
Our method is based on the Rice method and gives quite precise formulas. 
We derive upper bounds on the number of survivors in this algorithm and we propose a proper use of LGE.
\par 
Finally, we discuss one property of a general urns and balls problem and show a lower bound for a required number of rounds for a large class of distributed leader election protocols.
\par
\end{abstract}

\section{Introduction}

In \cite{DBLP:conf/mascots/JacquetMM13} Philippe Jacquet, Dimitris Milioris and Paul M{\"u}hlethaler introduced 
a novel energy efficient broadcast leader election algorithm, which they called,
in accordance with the popular fashion in those years, 
a Leader Green Election (LGE). 
This algorithm was also presented by P. Jacquet at the conference AofA'13.
\par
We will use the same model as in \cite{DBLP:conf/mascots/JacquetMM13}, namely we assume that
the communication medium is of the broadcast type and is
prone to collisions. We also assume that
the time is slotted. Each slot can be empty (the slot does not contain any burst),
collision (the slot contains at least two burst) or successful (the slot contains a single burst).
\par
During the investigation of efficiency of LGE algorithm we found a connection of the leader election problem with some  properties of the general "`urns and balls"' model. This connection is discussed in Section \ref{sec:lowerbound}.

\subsection{Short Description of LGE}

We will give a short description of a slightly simplified version of the LGE algorithm (for example,
authors of \cite{DBLP:conf/mascots/JacquetMM13} consider an arbitrary base of numeral systems, but we restrict our considerations only to base $3$, since some additional arguments, not presented in this paper, show that base-3 is an optimal choice for our purposes).\par

We assume that the broadcast medium has $N$ connected users (assume $N \approx 10^6$) and that
the number of contenders $n$ is always smaller or equal to $N$. 
We fix a number $p\in(0,1)$ and we assume that p is not close to one (e.g.
$p = 0.01$). We also fix a number $L = \BigO{\log \log N}$.

Each  contender $\omega$ selects independently a random number $g_\omega$ according to the geometric distribution with parameter $p$ (see next section for details). 
If $g_\omega \geq 3^{L+1}$ then we put $g_\omega = 0$. 
The number $g_\omega$ is written
\begin{equation}
\label{eq:base3exp}
  g_\omega = \sum_{k=0}^{L} b_k \cdot 3^k~, 
\end{equation}
where $b_k \in \{0,1,2\}$. We fix a function $f:\{0,1,2\}\to\{0,1\}^2$ by $f(0)=00$,
$f(1) = 01$ and $f(2)=10$, 
and define the transmission key $K_\omega$ for a contender $\omega$ as the concatenation
$$
  K_\omega = f(b_L) || f(b_{L-1}) || \ldots || f(b_1) || f(b_0) ~.
$$
Notice that lenght($K_\omega$) = $\BigO{\log\log N}$.
This key $K_\omega$ is used in the following algorithm played in discrete rounds:
\begin{algorithmic}[1]
\STATE candidate = \TRUE
\FOR {i=1 \TO lenght($K_\omega$)}
\IF {$K_\omega(i)$ = 1}
\STATE send a beep
\ELSE
	\STATE listen
	\IF{you hear a beep} 
		\STATE candidate = \FALSE
		\STATE exit loop
	\ENDIF
\ENDIF
\ENDFOR
\end{algorithmic}

The survivors of this algorithm are those contenders which at the end have the variable "candidate"
set to true. In \cite{DBLP:conf/mascots/JacquetMM13} authors propose to repeat this algorithm 
several times in order to reduce the number of survivors to 1. However we propose in this paper 
an another approach: we propose to use this algorithm only once (in order to reduce number 
of survivors to a small number) and then to use other leader election algorithm for final selection a leader.
 
\subsection{Mathematical Background}

The core of LGE algorithm is based on properties of extremal statistics 
of random variables with geometric distributions.  
Let us recall that a random variable $X$ has a geometric distribution 
with parameter $p \in [0,1]$ ($X \sim \Geo{p}$)
if $P[X=k] = (1-p)^{k-1}p$ for $k\geq 1$. 
In the first part of LGE, each user chooses independently a random variable with geometric
distribution with a fixed parameter $p$. The winners of this part of LGE are those users who select
a maximal number.  

\begin{definition}
A random variable $M$ has distribution $\MGeo{n}{p}$ if there are independent random variables 
$X_1,\ldots,X_n$ with distribution $\Geo{p}$ such that $$M = \max\{X_1,\ldots,X_n\}~.$$
\end{definition}

It is well known  (see e.g. \cite{Szpankowski:1990:YAB:78907.78914},
\cite{DBLP:journals/fuin/CichonK13}) that if $M \sim \MGeo{n}{p}$ then
$\E{M}$ = $\frac12 + \frac{\HarmonicN{n}}{\ln \frac{1}{1-p}} + P(n) + \BigO{\frac1n}$, 
where $P(n)$ is a periodic function with small amplitude and $\HarmonicN{n}$ is the $n^{\mathrm{th}}$ harmonic number.
Let us recall that $\HarmonicN{n} = \ln n + \gamma + \BigO{\frac1n}$, where $\gamma = 0.557\ldots$ is the Euler constant.

The distribution $\MGeo{n}{p}$ controls the number of time slots used in LGE algorithm.
More precisely, the LGE algorithm requires some upper approximation on the variable 
with the $\MGeo{n}{p}$ distribution. The next Lemma gives some upper bound for it.

\begin{lemma}
\label{lemma:boundNumberRounds}
Let $M \sim \MGeo{n}{p}$, $C>0$ and $Q = \frac{1}{1-p}$. Then
$$\Pr[M>C \frac{\ln n}{\ln Q}] \leq \frac{1}{n^{C-1}}~.$$
\end{lemma}

\begin{proof}
Let $q=1-p$. Let us recall that if $X \sim \Geo{p}$ and $k$ is an integer then  $\Pr[X>k] = q^k$.
Therefore $\Pr[M>k] \leq n q^k$, hence
$\Pr\left[M>C \frac{\ln n}{\ln Q}\right] \leq n q^{C \frac{\ln n}{\ln Q}} = \frac{1}{n^{C-1}}$.
\end{proof}

We introduce the next distribution which models the number of survivors in LGE algorithm.

\begin{definition}
A random variable $W$ has distribution $\WGeo{n}{p}$ if there are independent random variables 
$X_1,\ldots,X_n$ with distribution $\Geo{p}$ such that 
$$
  W = \card{\{k: X_k = \max\{X_1,\ldots,X_n\}\}}~.
$$
\end{definition}

\section{Probabilistic Propeties of LGE}

The formal analysis of LGE algorithm in \cite{DBLP:conf/mascots/JacquetMM13} is based 
on the Mellin transform. In this section, we use an approach based on Rice's method (see e.g. \cite{Knuth:1998:ACP:280635} and \cite{journals/tcs/FlajoletS95}). We shall derive formulas for expected number of survivors and probabilities  
for the number of survivors. By $W_{n,p}$ we denote a random variable with $\WGeo{n}{p}$ distribution.
  
\begin{theorem}
\label{theorem:winners}
Let $n \geq 2$, $p\in (0,1)$ and $q=1-p$. Let $W_{n,p} \sim \WGeo{n}{p}$ and $a\geq 1$. Then
$$
  \Pr[W_{n,p}=a] = \binom{n}{a} p^a  \sum_{b=0}^{n-a} \binom{n-a}{b} \frac{(-1)^b}{1-q^{a+b}} 
$$
and
$$
  \E{W_{n,p}}   = \frac{n p}{q} \sum_{b=0}^{n-1} \binom{n-1}{b} \frac{(-1)^b}{1-q^{b+1}}~.
$$
\end{theorem}

\begin{proof} 
Let us fix $n\geq 2$, $p \in (0,1)$ and $q=1-p$.
Let $X_1,\ldots,X_n$ be independent random variables with distribution $\Geo{p}$ and
let
$$
  A_{n;k,a} =(\max\{X_1,\ldots,X_n\}=k) \land (\card{\{ i : X_i = k\} = a})~.
$$
Then
$[W_{n,p}=a] =  \bigcup_{k\geq 1} A_{n;k,a}$
and 
$\Pr[A_{n;k,a}] = \binom{n}{a} (q^{k-1}p)^a (1-q^{k-1})^{n-a}$.
Therefore,
\begin{gather*}
  \Pr[W_n=a] = \sum_{k\geq 1} \binom{n}{a} (q^{k-1}p)^a(1-q^{k-1})^{n-a} = \\
	\binom{n}{a} p^a  \sum_{k\geq 0} q^{ka}(1-q^k)^{n-a} = 
	\binom{n}{a} p^a  \sum_{k\geq 0}\sum_{b=0}^{n-a} \binom{n-a}{b}(-1)^b q^{kb}q^{ka} = \\ 
	\binom{n}{a} p^a  \sum_{b=0}^{n-a} \binom{n-a}{b}(-1)^b \sum_{k\geq 0}q^{k(b+a)} = 
	\binom{n}{a} p^a  \sum_{b=0}^{n-a} \binom{n-a}{b} \frac{(-1)^b}{1-q^{a+b}} ~,
\end{gather*}
so the first part of the Theorem is proved.
Next we have
\begin{gather*}
\sum_{a=1}^{n} a \Pr[A_{n;k,a}] = 
%\sum_{a=1}^{n} a \binom{n}{a} (q^{k-1}p)^a (1-q^{k-1})^{n-a} = \\
n \sum_{a=1}^{n} \binom{n-1}{a-1} (q^{k-1}p)^a (1-q^{k-1})^{n-a} = \\
n q^{k-1}p \sum_{a=1}^{n} \binom{n-1}{a-1} (q^{k-1}p)^{a-1} (1-q^{k-1})^{(n-1)-(a-1)} = \\
n q^{k-1}p \sum_{b=0}^{n-1} \binom{n-1}{b} (q^{k-1}p)^{b} (1-q^{k-1})^{(n-1)-b} = \\
n q^{k-1}p (q^{k-1}p + 1- q^{k-1})^{n-1} =
%n q^{k-1}p (q^{k-1}(p-1) + 1)^{n-1} = \\
n q^{k-1}p (1-q^k)^{n-1}~.
\end{gather*}
Therefore, for fixed $n$, we have
\begin{gather*}
\sum_{k\geq 1}\sum_{a=1}^{n} a \Pr[A_{k,a}] = 
\sum_{k\geq 1} n p q^{k-1} \sum_{b=0}^{n-1} \binom{n-1}{b}(-1)^b q^{k b} = \\
n p \sum_{b=0}^{n-1} \binom{n-1}{b}(-1)^b \sum_{k\geq 1}  q^{k-1} q^{k b} = 
\frac{n p}{q} \sum_{b=0}^{n-1} \binom{n-1}{b}(-1)^b \sum_{k\geq 1}  q^k q^{k b} = \\
\frac{n p}{q} \sum_{b=0}^{n-1} \binom{n-1}{b}(-1)^b \sum_{k\geq 1}  q^{k (b+1)} = 
\frac{n p}{q} \sum_{b=0}^{n-1} \binom{n-1}{b}(-1)^b q^{b+1}\sum_{k\geq 0}  (q^{b+1})^k = \\
\frac{n p}{q} \sum_{b=0}^{n-1} \binom{n-1}{b} \frac{(-1)^b q^{b+1}}{1- q^{b+1}}~.
\end{gather*}

Since we assumed that $n\geq 2$, we have
\begin{gather*}
\frac{n p}{q} \sum_{b=0}^{n-1} \binom{n-1}{b} \frac{(-1)^b q^{b+1}}{1- q^{b+1}} = 
\frac{n p}{q} \sum_{b=0}^{n-1} \binom{n-1}{b} (-1)^b \frac{ (q^{b+1}-1) +1}{1- q^{b+1}} = \\ 
-\frac{n p}{q} \sum_{b=0}^{n-1} \binom{n-1}{b} (-1)^b + 
\frac{n p}{q} \sum_{b=0}^{n-1} \binom{n-1}{b} (-1)^b \frac{1}{1- q^{b+1}} = \\
-\frac{n p}{q} (-1+1)^{n-1} + 
\frac{n p}{q} \sum_{b=0}^{n-1} \binom{n-1}{b}  \frac{(-1)^b}{1- q^{b+1}} =
\frac{n p}{q} \sum_{b=0}^{n-1} \binom{n-1}{b}  \frac{(-1)^b}{1- q^{b+1}}~. 
\end{gather*}

\end{proof}

From Theorem \ref{theorem:winners} we obtain the following equality $\Pr[W_{n,p}=1]$ = $n p  \sum_{b=0}^{n-1} \binom{n-1}{b} \frac{(-1)^b}{1-q^{1+b}}$.
Therefore, we have the following nice equality
$$
\E{W_{n,p}} = \frac{1}{1-p}\Pr[W_{n,p}=1] ~.
$$ 
\textbf{Remark} Quite recently we learned that Theorem \ref{theorem:winners} and part of results from the next subsection has been proved in \cite{kirschenhofer1996}. 
Due to the completeness of arguments we decided to leave the proof in this paper. 
Our new contribution in this section is the  Theorem \ref{theorem:winnersBound}.

\subsection{Approximations}

Let us fix the number $p\in(0,1)$ and let $q = 1-p$.
Let $f_a(z) = \frac{1}{1-q^{a+z}}$. We shall consider complex variable functions $f_a$ for such indexes $a$ 
which are integers such that $a\geq 1$.
Notice that the function $f_a$ has singularities at points from the set $\{\zeta_{a,k} : k\in\ZZ\}$,
where $\zeta_{a,k} = -a + \frac{2 k \pi \mathbf{i}}{\ln(q)}$. 
The function $f_a$ is periodic with period $2\pi\mathbf{i}/\ln(q)$,
has single poles at points $\zeta_{a,k}$ and
$$
  \Res{f_a(z)}{z}{\zeta_{a,k}} = \frac{-1}{\ln q} ~.
$$
It is easy to check  that $\lim_{x\to\infty} |f_a(x + \mathbf{i}y)| = 1$ and
$\lim_{x\to-\infty} |f_a(x + \mathbf{i}y)| = 0$ for each fixed $y \in \RR$.

Let $K_n(s) = \frac{n!}{s(s-1)\cdots(s-n)}$. Notice that if $n \geq 1$ then $|K_n(s)| = \BigO{\frac{1}{|s|^2}}$ 
as $|s|$ grows to infinity.  Also notice that if $a>0$ is an integer, then 
$K_n(-a) =  (-1)^{n+1} \frac{1}{a} \binom{n}{a}^{-1}$.
Notice also that the sets of singularity points of functions $f_a$ and $K_n$ are disjoint.
This fact greatly simplifies the analysis of the singular points of the product of these functions

\begin{lemma} 
\label{lemma:afterRice}
If $m\geq 1$, $a\geq 1$ and $q \in (0,1)$ then
$$
  \sum_{b=0}^{m} \binom{m}{b} \frac{(-1)^b}{1- q^{a+b}} = (-1)^{m} \frac{1}{\ln q} \sum_{k\in\ZZ} K_m(\zeta_{a,k})~.
$$
\end{lemma} 

\begin{proof}
Rice's integrals summation method (see \cite{Knuth:1998:ACP:280635}) is based on the formula
$$
  \sum_{b=0}^{m} \binom{m}{b}(-1)^b g(b) = 
	\frac{(-1)^m}{2 \pi i} \oint_{\mathcal{C}} g(s) K_m(s) ds ~,
$$  
where $g$ is analytic in a domain containing $[0,+\infty)$ and
$\mathcal{C}$ is a positively oriented closed curve that lies in the domain of analyticity of $g$ and
encircles the real interval $[0, m]$. 

We use Rice' formula for functions $f_a$. 
Notice that
$$
\frac{1}{2 \pi i} \oint_{\mathcal{C}} f_a(s) K_m(s) ds =  \sum_{k=0}^{m} \Res{f_a(z)K_m(z)}{z}{k}~.
$$

Let $C_k$ be the positively oriented square with corners at points
$\pm \eta_{q,k} \pm \eta_{q,k} \mathbf{i}$, where $\eta_{q,k} = (2k+1)\pi/\ln q$. We consider such $k$ that
$|\eta_{q,k}|>m$. For such $k$ the interval $[0,m]$ lies inside the square $C_k$. 
The mentioned before Lemma \ref{lemma:afterRice} properties of the function $f_a$ 
(periodicity and boundedness on horizontal lines not crossing singular points) and the kernel function $K_m$ imply that
$$
  \lim_{k\to\infty} \oint_{C_k} f_a(s) K_m(s) ds = 0~, 
$$ 
from which we deduce that
\begin{gather*}
  \sum_{k\in\ZZ} \Res{f_a(z) K_m(z)}{z}{\zeta_{a,k}} + 
	\sum_{k=0}^{m} \Res{f_a(z)K_m(z)}{z}{k} = 0~.
\end{gather*}
Therefore,
\begin{gather*}
\sum_{b=0}^{m} \binom{m}{b} \frac{(-1)^b}{1- q^{a+b}} = 
  (-1)^{m+1} \sum_{k\in\ZZ} \Res{f_a(z) K_m(z)}{z}{\zeta_{a,k}} = \\ 
	(-1)^{m+1} \sum_{k\in\ZZ} \Res{f_a(z)}{z}{\zeta_{a,k}} K_m(\zeta_{a,k}) = 
	(-1)^{m+1} \sum_{k\in\ZZ} \frac{-1}{\ln q} K_m(\zeta_{a,k}).
\end{gather*}
\end{proof}

\begin{lemma}
\label{lemma:kernelfunction}
Suppose that $a>0$ is an integer and that $b\in\CC$. Then
$$
  K_m(-a + b) =  
	\frac{(-1)^{m+1}}{a \binom{a+m}{a}} \cdot \frac{1}{\prod_{j=a}^{m+a} (1 - \frac{b}{j})} ~.
$$
\end{lemma}
\begin{proof}
Directly from the definition of the kernel function $K_m$ we have
\begin{gather*}
K_m(-a + b) = m! \prod_{j=0}^{m} \frac{1}{-a+b-j} = 
(-1)^{m+1} m! \prod_{j=0}^{m} \frac{1}{a+j-b} = \\
(-1)^{m+1} m! \prod_{j=a}^{m+a} \frac{1}{j-b} = 
(-1)^{m+1} m! \prod_{j=a}^{m+a} \frac{1}{j(1-\frac{b}{j})}=\\
(-1)^{m+1} m! \frac{(a-1)!}{(m+a)!} \prod_{j=a}^{m+a} \frac{1}{(1-\frac{b}{j})} ~.
\end{gather*}
\end{proof}

\noindent
The next Lemma follows directly from Theorem \ref{theorem:winners}, Lemmas \ref{lemma:afterRice} and  \ref{lemma:kernelfunction}:  
\begin{lemma}
\label{lemma:probability}
If $n>a$ then
$$
  \Pr[W_{n,p}=a] = \frac{p^a}{a \ln \frac1q} \left(1+ \sum_{k\in\ZZ\setminus\{0\}} \frac{1}{\prod_{j=a}^{n} 
	(1 - \frac{2 k \pi \mathbf{i}/\ln q}{j})} \right) ~,
$$
where $q = 1-p$.
\end{lemma}

\begin{theorem} 
\label{thm:finalappro}
If $0< a < n$ then 
$\Pr[W_{n,p}=a] = \frac{p^a}{a\ln(Q)}+r_n$, where $|r_n|<\frac{(a+1)^2}{12a}p^a \ln(Q)$,
where $Q = \frac{1}{1-p}$. 
\end{theorem}

\begin{proof}
Let $\eta_k = \frac{2\pi k \mathbf{i}}{\ln q}$, where $q=1-p$.
Notice that
\begin{gather*}
	\left|\prod_{j=a}^{n}(1-\frac{\eta_k}{j})\right|^2 =  
	\prod_{j=a}^{n}\left(1 + \frac{|\eta_k|^2}{j^2}\right) \geq 
	\prod_{j=a}^{a+1}\left(1 + \frac{|\eta_k|^2}{j^2}\right) \geq
	\left(1+\frac{|\eta_k|^2}{(a+1)^2}\right)^2 ~.
\end{gather*}
Therefore,
\begin{gather*}
\left| \sum_{k\in\ZZ\setminus\{0\}}\prod_{j=a}^{n}\frac{1}{1-\frac{\eta_k}{j}} \right| \leq 
2 \sum_{k=1}^{\infty} \frac{1}{1+\frac{|\eta_k|^2}{(a+1)^2} } \leq
2(a+1)^2 \sum_{k=1}^{\infty} \frac{1}{|\eta_k|^2} = \\
\frac{(a+1)^2(\ln q)^2}{2\pi^2} \sum_{k=1}^{\infty} \frac{1}{k^2} =
\frac{(a+1)^2(\ln q)^2}{12} ~,
\end{gather*}
so the conclusion follows from Lemma \ref{lemma:probability}.
\end{proof}
 
Let  us fix $p \in (0,1)$, let $Q =\frac{1}{1-p}$. We put
$$
\phi_p(a) = \frac{p^a}{a \ln{Q}} + \frac{(a+1)^2}{12 a} p^a \ln{Q}.
$$
Notice that $\Pr[W_{n,p}=a] \leq \phi_p(a)$.
\begin{theorem}
\label{theorem:winnersBound}
 $\Pr[W_n \geq k] < \frac{\phi(k)}{1-2p}$
\end{theorem}
\begin{proof}
It can be observed that $\frac{\phi_p(a+1)}{\phi_p(a)} < 2 p$.  
Therefore,
\begin{gather*}
 \Pr[W_n\geq k]=\sum_{a=k}^n\Pr[W_n=a]<\sum_{a=k}^\infty \phi(a)<\frac{\phi(k)}{1-2p} ~.
\end{gather*}
\end{proof}

\subsection{Discussion}

Let us observe that formulas from Theorem \ref{thm:finalappro} do not depend on the number $n$. 
However, small fluctuations (which are very interesting from theoretical point of view) are hidden inside the error term, which can be observed on the Fig. \ref{fig:onewinner}.  

\begin{figure}[ht]
\label{fig:onewinner}
\centering
\includegraphics[width=0.75\textwidth]{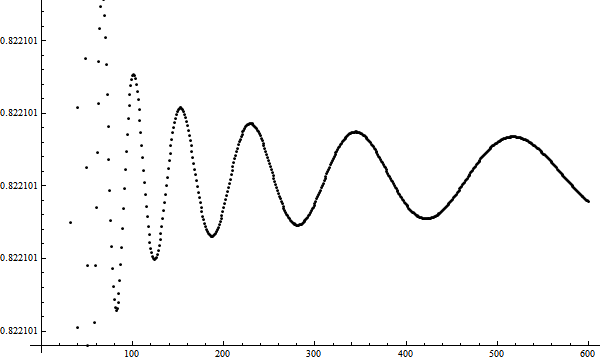}
\caption{Plot of $\Pr[W_{n,\frac13}=1]$ for $n=1,\ldots,600$.}
\end{figure}

This  practical independence of the number $n$ of nodes on the number of survivors is very interesting. 
However, the number $n$ has an influence on the required number of rounds in LGE.
This number may be controlled by Lemma 
\ref{lemma:boundNumberRounds}:  from this lemma we deduce that if 
$X \sim \MGeo{n}{p}$ then $\Pr[X > (\ln 10^{20} + \ln n)/\ln(Q)] < 10^{-20}$ (where $Q = 1/(1-p)$), 
and hence from a practical point of view it is negligible.
This implies that (see \cite{DBLP:conf/mascots/JacquetMM13} for details) the LGE algorithm should run
$2 \cdot \left\lceil \log_3 \left(\frac{1}{\ln(Q)}(\ln n+  \ln(10^{20})\right)\right\rceil$ rounds in order 
to ensure that its probabilistic properties are controlled by the distribution $\mathrm{WGeo}$ with probability at least
$1-10^{-20}$.
%the number of survivors to $10$ or less with a probability of order $1 -10^{-19}$ 

From Theorem \ref{thm:finalappro} we  deduce that $\Pr[W_{n,p}=1]= 1- \frac{p}{2} + \BigO{p^2}$
and $\Pr[W_{n,p}=2]= \frac{p}{2} + \BigO{p^2}$.
From these formulas we deduce that the probability of failure of one phase of LGE is quite large. However, notice
that from Theorem \ref{theorem:winnersBound} we get 
$\Pr[W_{n,0.01}>10] \approx 1.006 \cdot 10^{-19}$.
Therefore, the LGE algorithm may be used for quick reduction of potential leaders to a small subgroup. 
We see that if we use this algorithm with parameter $p = \frac{1}{100}$, then with
probability at least $1- 10^{-19}$, the number of survivors will be less or equal  $10$.
The survivors may then take part in another algorithm (e.g. in an algorithm based on paper 
\cite{DBLP:journals/dm/Prodinger93} 
or in algorithm based on paper \cite{DBLP:journals/combinatorics/JansonS97}, \cite{GLandHP2009}), 
which deals better with small sets of nodes,  in order to select a leader with high and controllable 
probability.

\section{Lower Bound}
\label{sec:lowerbound}

In the previous section we recalled that the LGE algorithm should use $\BigO{\ln\ln(n)}$ rounds in order to achieve high
effectiveness. In this section we prove a general result confirming that this bound is near to an optimal.
We use a method applied by D. E. Willard in \cite{Willard:1986} for an analysis of resolution protocols
in a multiple access channel.

Let us consider  a system $(U_i)_{i=1,\ldots,L}$ of $L$ urns and let us fix a number $n$.
We consider a process of throwing an arbitrary number $Q \in \{2,\ldots,n\}$ of balls into 
these urns. We assume that all balls are thrown independently and that the probability that
the ball is thrown into $i$th urn is equal $p_i$.
This process is fully described by the vector $\vec{p}$ of probabilities from the simplex
$\Sigma_L = \{(p_1,\ldots,p_L) \in [0,1]^L: p_1+\ldots+p_L = 1\}$ and the number $Q$ of balls.

The most broadly studied model of urns and balls is the uniform case, i.e. the case when 
$\vec{p} = (\frac1L,\ldots,\frac1L)$. However, in several papers 
(see e.g. \cite{Flajolet:1992}, \cite{boneh1997coupon}) one can find some results for the general case.
In this section we are interested in the existence of at least one singleton, i.e. in the existence of an urn $U_i$
with precisely one ball. The problem of estimation of the number of singletons was quite recently analyzed in 
\cite{EJP699}.

Let  $S_{\vec{p},Q}$ denote the event "there exists at least one urn with a single ball"
and let $S_{\vec{p},Q,i}$ denote the event "there is exactly one ball in $i$th urn". Then, 
$\Pr[S_{\vec{p},Q,i}] = Q p_i (1-p_i)^{Q-1}$ and $S_{\vec{p},Q} = \bigcup_{i=1}^{L} S_{\vec{p},Q,i}$,
therefore,
$\Pr[S_{\vec{p},Q}] \leq Q \sum_{i=1}^{L} p_i (1-p_i)^{Q-1}$.

Let us assume that the number $Q$ of balls is unknown but it is bounded by a number $n$.
We are going to show that if the number $n$ is sufficiently large compared to $L$, then
there is no $\vec{p} \in \Sigma_L$ which will guarantee the existence of singleton with a high probability 
for arbitrary $Q$ from $\{2,\ldots,n\}$.
More precisely, let
$$
  \MSP{L}{n} = \max_{\vec{p}\in\Sigma_L} \min_{2\leq Q \leq n} \Pr[S_{\vec{p},Q}] ~.
$$
(term MSP is an abbreviation of "Maximal Success Probability").
 
\begin{theorem}
\label{thm:main}
For arbitrary $L\geq 1$ and $n \geq 2$, we have
$$\MSP{L}{n} < \frac{L-1}{\HarmonicN{n} -1}~.$$
\end{theorem}

\begin{proof}
Let us observe that if $\vec{p}\in\Sigma_L$ is such that for some $i$ we have $p_i = 1$ and $Q\geq 2$, then
$\Pr[S_{\vec{p},Q}]=0$, so  $\min_{2\leq Q \leq n} \Pr[S_{\vec{p},Q}] = 0$. Hence, we may consider only
such $\vec{p}\in\Sigma_L$ that $p_i<1$ for each $i=1,\ldots,L$.
 
Let us fix  the number $L$ of urns and let us consider the following function (this is the trick which we borrow from
\cite{Willard:1986}):
$$
  f(\vec{p}) = \sum_{Q=2}^{n} \frac{\Pr[S_{\vec{p},Q}]}{Q}~. 
$$
Then we have
\begin{gather*}
	f(\vec{p}) \leq \sum_{Q=2}^{n} \sum_{i=1}^{L}\frac{\Pr[S_{\vec{p},Q,i}]}{Q} =  
                \sum_{Q=2}^{n} \sum_{i=1}^{L}p_i (1-p_i)^{Q-1} \leq 
								\sum_{i=1}^{L} \sum_{Q=2}^{\infty} p_i (1-p_i)^{Q-1} = \\
	\sum_{i=1}^{L} p_i (1-p_i) \frac{1}{1-(1-p_i)} = \sum_{i=1}^{L} (1-p_i) = 
	L - \sum_{i=1}^{L} p_i = L-1~. 	
\end{gather*}
On the other side, let $p^* = \min\{\Pr[S_{\vec{p},Q}]: 2\leq Q \leq n\}$. Then we have

\begin{gather*}
	f(\vec{p}) \geq \sum_{Q=2}^{n} \frac{p^*}{Q} = 
                  p^* \sum_{Q=2}^{n} \frac{1}{Q}  = p^* (\HarmonicN{n} - 1)~. 	
\end{gather*}
Therefore, we have
$$
  p^* (\HarmonicN{n} - 1) \leq f(\vec{p}) < L-1~.
$$
Hence, if we take $Q^*$ such that $\Pr[S_{\vec{p},Q^*}] = p*$, then
$$
  \Pr[S_{\vec{p},Q^*}] <\frac{L-1}{\HarmonicN{n}-1}~,
$$
so 
$$
  \min_{2\leq Q \leq n} \Pr[S_{\vec{p},Q}] < \frac{L-1}{\HarmonicN{n}-1}
$$
for arbitrary $\vec{p} \in \Sigma_n$.  
\end{proof}

\begin{corollary}
\label{cor:bound01}
If $1 \leq L\leq  \frac12 \ln n + \frac{1+\gamma}{2}$ then   $\MSP{L}{n} < \frac12$.
\end{corollary}

\begin{corollary}
\label{cor:bound02}
If $n \geq \exp(2 L - (1+\gamma))$  then   $\MSP{L}{n} < \frac12$.
\end{corollary}

\begin{proof}
Both proofs follow directly from Theorem \ref{thm:main} and the inequality $\HarmonicN{n} \geq \ln(n) + \gamma$. 
\end{proof}

\subsection{Application to Leader Election Problem}

Let us consider any oblivious leader election algorithm in which at the beginning each station selects randomly and independently
a sequence of bits of length $M$, and later this station use the sequence 
in the algorithm in a deterministic way. 
Let $n$ denote the upper bound on the number of stations taking part in this algorithm 
and let $b_i$ denote the sequence of bits chosen by the $i$th station.
Observe that if for each $i$ there is $j\neq i$ such that $b_i = b_j$, then the algorithm must fail. 
Hence, success is possible only if there is a singleton  in choices made from the space $\{0,1\}^M$ 
of all possible sequences of bits. When we use Corollary \ref{cor:bound01} with $L = 2^M$, then we deduce
that if $M \leq \log_2\left(\frac12 \ln n + \frac{1+\gamma}{2}\right)$ then the probability that the considered algorithm
chooses a leader is less than  $\frac12$. We may say that $\log_2(\frac12 \ln n)$ random bits are too few for distinguishing 
an arbitrary collection of $\leq n$ objects with a high probability.  

\section*{Acknowledgment}
The authors wish to thank the anonymous reviewers for their valuable comments and for drawing attention to that a large part of the results from Section 2 of this paper has already been proven in \cite{kirschenhofer1996} and that the are new results about the asymmetric leader election algorithm (e.g. \cite{GLandHP2009}) not mentioned in a previous version of our paper.
 
\bibliographystyle{abbrvnat}
\bibliography{Green}
\label{lastpage}
\end{document}